%% file: main.tex
\documentclass[letterpaper, 10 pt, conference]{ieeeconf}
\IEEEoverridecommandlockouts
\makeatletter
\def\NAT@parse#1{}
\makeatother

\usepackage{graphics}      
\usepackage{epsfig}        
\usepackage{amsmath}
\usepackage{amssymb}
\usepackage{mathtools}
\usepackage{bm}
\usepackage{multirow}
\usepackage{color}
\usepackage{paralist}
\usepackage[space]{cite}
\usepackage{graphicx}
\usepackage{url}
\usepackage[hidelinks]{hyperref}
\usepackage{cleveref}
\usepackage{caption}
\usepackage{subcaption}
\usepackage{bm}
\usepackage{float}
\usepackage{tikz}
\usepackage{algorithm}
\usepackage{algorithmicx}
\usepackage{algpseudocode}
\usepackage{siunitx}

\usepackage{xcolor}

\input{setting.tex}

\title{\LARGE \bf On Word-of-Mouth and Private-Prior Sequential Social Learning}

\author{Andrea {Da Col}, Cristian R. Rojas and Vikram Krishnamurthy
\thanks{}
\thanks{This work has been partially supported by the Swedish Research Council under contract number 2023-05170, and by the Wallenberg AI, Autonomous Systems and Software Program (WASP) funded by the Knut and Alice Wallenberg Foundation, and the US National Science Foundation grant CCF-2112457. Andrea Da Col and Cristian R. Rojas are with the Division of Decision and Control Systems, KTH Royal Institute of Technology, 100 44 Stockholm, Sweden. Vikram Krishnamurthy is with the School of Electrical and Computer Engineering, Cornell University, Ithaca, NY, 14853, USA.
E-mails: {\tt\small andreadc@kth.se}, {\tt\small crro@kth.se}, {\tt \small vikramk@cornell.edu}.
}%
}

\begin{document}

\maketitle
\thispagestyle{empty}
\pagestyle{empty}


\input{abstract}


\input{introduction}


\input{problem_formulation}  


\input{theoretical_analysis}


\input{numerical_results}


\input{conclusions}


\bibliography{References}

\end{document}

%% file: setting.tex
\newtheorem{lemma}{\bf{Lemma}}
\newtheorem{theorem}{\bf{Theorem}}
\newtheorem{remark}{\bf{Remark}}


%% file: abstract.tex

\begin{abstract}
Social learning constitutes a fundamental framework for studying interactions among rational agents who observe each other’s actions but lack direct access to individual beliefs. 
This paper investigates a specific social learning paradigm known as Word-of-Mouth (WoM), where a series of agents seeks to estimate the state of a dynamical system. 
The first agent receives noisy measurements of the state, while each subsequent agent relies solely on a degraded version of her predecessor’s estimate. 
A defining feature of WoM is that the final agent’s belief is publicly broadcast and subsequently adopted by all agents, in place of their own. 
We analyze this setting theoretically and through numerical simulations, noting that some agents benefit from using the belief of the last agent, while others experience performance deterioration.
\end{abstract}


%% file: introduction.tex

\section{Introduction}
Use of synthetic datasets has become a standard practice in several machine learning workflows, enabling model training and validation, as well as data augmentation \cite{nikolenko2021synthetic}.
Notably, the repeated use of identical information makes this procedure prone to the problem of data incest \cite{vikram2003incest}. With the recent rise in popularity of generative machine learning models \cite{diffusion2024, zhao2023survey}, renewed interest and concern have emerged regarding the issues of misinformation propagation and mode collapse \cite{marchi2024heatdeathgenerativemodels}.
By incorporating synthetic data in the public web, generative models would actively update their training sets for future campaigns. In the worst case, injecting gibberish information can pollute this shared memory, leading to significant model degradation \cite{biggio2012poisoning}. Given the imminent large-scale deployment of generative artificial intelligence, a solid theoretical framework for understanding the effects of data incest is needed.

In this paper, we analyze the online learning dynamics of a group of decision-makers that are fed back their own outputs as part of future observations. Our study is conducted under the powerful lens of social learning, a fundamental mathematical framework for modeling interactions between social sensors \cite{AO11,Ban92,BHW92,BMS20,KH15,KP14,Say14b}. Widely used in economics and social sciences to model financial markets and network dynamics, social learning can effectively express how rational agents learn \cite{Cha04}. 
In social learning, agents estimate an underlying state of nature by using their private beliefs and the observed actions of their predecessors. This process follows a recursive Bayesian framework, wherein each agent applies Bayes' rule to infer the state based on the others' actions, each of which is itself the outcome of a Bayesian update. Such nested Bayesian inference can give rise to nontrivial phenomena. Notably, in finite state, observation and action spaces, agents often converge to an information cascade: after a finite number of time steps, they disregard private observations and blindly replicate past actions \cite{jain2025interacting}. 

Classical social learning models assume that each agent retains her private belief from the previous epoch, and updates it upon observing new actions of other agents.  In our discussion, the term private-prior (PP) is used to address this setup, where an agent is unaware of the other agents' beliefs.
A second architecture, inspired by Word-of-Mouth learning \cite{Vives01,Viv97}, is considered as an alternative to PP. 
In this setup, the first agent has access to external state observations, whereas subsequent ones can only rely on a function of their predecessors' actions. 
A crucial aspect of this framework is that the final agent dictates her personal belief to all the others, ultimately establishing a form of shared knowledge. 

Building on our recent work~\cite{slowConvergenceSL}, we study social learning among interacting Kalman filter agents that observe past actions in Gaussian noise. The agents estimate a scalar Gaussian state governed by linear dynamics, with their actions corresponding to their state estimates. 
We present a formal asymptotic analysis for the propagation of uncertainty in the PP and WoM settings, highlighting their similarities and differences.
This 
is nontrivial and requires careful analysis of the asymptotic properties of a discrete-time Riccati equation.

Unlike classical information cascades, we point out that WoM social learning affects the agents unevenly: while some of them experience degraded estimation performance, others benefit from shared information. In particular, WoM will produce more accurate estimates (smaller variance) for the last agent in the sequence, compared to PP learning. 
The opposite is true for the first agent of the group.
These last results can be interpreted in the context of a hierarchical social network.
Suppose a junior worker obtains noisy  measurements of the state and relays its noisy estimates (recommendations) to more senior workers,  eventually reaching  management. In the WoM framework, the manager sets the prior (corporate strategy) for the next sequence of  workers' recommendations.  In contrast, in the PP framework, each worker maintains their  own prior to update their recommendation. The implication is that a WoM manager is always better informed (has lower variance) than a PP manager when the underlying state evolves as a linear Gaussian system, despite some of the workers in the PP system being better informed.


Our contributions are threefold:

\vspace{-1mm}
\begin{itemize} \parsep 0pt \itemsep 0pt
\item We formulate a theoretical framework for Word-of-Mouth social learning in dynamical systems.
\item We characterize the stationary variance of the one-step-ahead prediction error as the fixed point of a Riccati equation, in both setups. Existence and uniqueness of this fixed point are established, and convergence is proved in the two-agent case.
\item We provide extensive numerical simulations comparing the two setups, validating the theoretical findings.
\end{itemize}

The paper is organized as follows: the private-prior and Word-of-Mouth frameworks are formally presented in Section~\ref{sec:PROBLEM FORMULATION}. A theoretical analysis 
of both setups is provided in Section~\ref{sec:THEORETICAL ANALYSIS}. Section~\ref{sec:EXPERIMENTS} complements the theoretical findings with numerical examples. Section~\ref{sec:CONCLUSIONS} concludes the paper.

\medskip
\textit{Notation:} Denote by $\mathbb{N}$ and $\mathbb{R}$ the sets of natural and real numbers, respectively. Let $\mathbb{R}_{>0}$ be the set of positive real numbers. Sequences are denoted by $(x_k)$.
With $\tilde x$ we indicate a realization for the random variable $x$. We use boldface fonts to denote exponents and normal fonts for superscripts. $\mathbb{E}$ is the expected value, and $\textrm{MSE}$ the mean squared error.

%% file: problem_formulation.tex

\section{Problem Formulation}\label{sec:PROBLEM FORMULATION}
A set $\mathcal{I}=\{1, \ldots, m\}$ of 
decision-makers aims at learning a time-dependent state of nature from private online measurements. 
Indexed by $i\in\mathcal{I}$, the agents are interconnected in a serial structure and are only allowed to operate sequentially: one cannot take actions before her predecessor.
In addition, they measure inherently different quantities: while the first agent can directly access the unknown state (contaminated with noise), each subsequent agent can only observe (in noise) the estimate produced by the one before her. Private observations are used by each agent to update her own belief of the unknown state, in a Bayesian way.

\subsection{Model and Agent Description}\label{subsec:MODEL SELECTION}
We consider a data-generating mechanism governed by the following first-order autoregressive dynamics:
\begin{equation}\label{eq:AUTOREGRESSIVE DYNAMICS}
    x_k = a x_{k-1} + w_k, \qquad k\in\mathbb{N},
\end{equation}
where $x_k, w_k \in\mathbb{R}$ are the state of the system and the process noise at time $k$, respectively. Let $x_0$ and $(w_k)$ be mutually independent Gaussian distributed random variables, such that $x_0\sim\mathcal{N}(\hat{x}_0,p_0)$ and $w_k\sim\mathcal{N}(0,q)$. We assume that the state dynamics are asymptotically stable, i.e., $a\in(-1,1)$, and that the agents of $\mathcal{I}$ can measure changes in the state according to the following set of output equations:
\begin{equation}\label{eq:OBSERVATION MODELS}
    y_k^{i} = x_k + n_k^{i}, \qquad k\in\mathbb{N}, \ i \in \mathcal{I},
\end{equation}
where $y_k^{i}, n_k^{i}\in\mathbb{R}$ are the observed output and the measurement noise of agent $i$ at time $k$, respectively. 
The Gaussian process $(n_k^{i})$, where $n_k^{i}\sim\mathcal{N}(0,r_k^{i})$, is assumed uncorrelated in time and independent of $x_0$ and $(w_k)$. 
With \eqref{eq:OBSERVATION MODELS} we express the quality of information that agent $i$ can access, which depends on the way agents are interconnected. 
Finally, we assume that all densities are non-degenerate, namely that $p_0, q \in\mathbb{R}_{>0}$, and $r_k^{i}\in\mathbb{R}_{>0}$ for every $k\in\mathbb{N}$, $i\in\mathcal{I}$. 

An agent $i$ implementing \eqref{eq:AUTOREGRESSIVE DYNAMICS} and \eqref{eq:OBSERVATION MODELS} can optimally estimate $x_k$ from her local observations $\{ \tilde{y}_1^{i}, \ldots,  \tilde{y}_k^{i} \}$ in the least squares sense \cite[Chapter\ 5]{simon2006optimal} by using the Kalman filter: 
%
\begin{equation}\label{eq:KALMAN UPDATE}
    \begin{aligned}
        & p_{k|k-1}^{i} = a^2 p_{k-1|k-1}^{i} + q, \\
        & \hat x_{k|k-1}^{i} = a \hat x_{k-1|k-1}^{i}, \\
        & \alpha_k^{i} = p_{k|k-1}^{i}/(p_{k|k-1}^{i} + r_k^{i}), \\
        & p_{k|k}^{i} = p_{k|k-1}^{i}(1 - \alpha_k^{i}), \\
        & \hat x_{k|k}^{i} = \hat x_{k|k-1}^{i} + \alpha_k^{i} (\tilde{y}_k^{i} - \hat x_{k|k-1}^{i} ),  
    \end{aligned}
\end{equation}
where $\hat{x}_{k|k-1}^{i}$ ($\hat{x}_{k|k}^{i}$) and $p_{k|k-1}^{i}$ ($p_{k|k}^{i}$) denote the predictive (posterior) mean and variance of agent $i$, respectively.
Here we set $\hat{x}_{0|0}^{i} = \hat{x}_{0}$ and $p_{0|0}^{i} = p_0$, for all $i\in\mathcal{I}$. 
Since $p_{k|k-1}^{i}$, $\alpha_k^{i}$ and $p_{k|k}^{i}$ do not depend on data, we will assume them to be publicly available to every agent.

Except for the first agent, the other ones with $i\in\mathcal{I}\setminus\{1\}$ can only observe a noisy version of $\hat{x}_{k|k}^{i-1}$, at each time $k$.
We denote by $(v_k^{i})$ the additive white Gaussian noise injected immediately before $i$, where $v_k^{i}\sim\mathcal{N}(0, s^{i})$ and $s^{i}\in\mathbb{R}_{>0}$. 
We also assume that $(v_k^{i})$ is independent of the other random variables, namely $x_0$, $(w_k)$, and $(v_k^j)$ for $i\neq j$.

\begin{remark}
    \textit{It is important to understand that $k$ stays fixed until every agent $i$ updates her estimate. In this regard, we are considering a system with two time scales, indexed by $k\in\mathbb{N}$ and $i\in\mathcal{I}$, respectively. 
    While the first one regulates the state dynamics in \eqref{eq:AUTOREGRESSIVE DYNAMICS}, the latter models the sequential computation and exchange of posterior estimates between the agents according to \eqref{eq:KALMAN UPDATE}. 
    }
\end{remark}

Two setups will be addressed in this paper: one where each $i\in\mathcal{I}$ updates her private prior via a Kalman filter, and one where all agents use a common prior, that is set by the $m$-th agent in a Word-of-Mouth fashion. 

\subsection{Private-Prior Setup}\label{subsec:PRIVATE PRIOR}

\begin{figure*}[t!]
    \vspace{0.4 em}
    \centering
    \includegraphics[width=\textwidth]{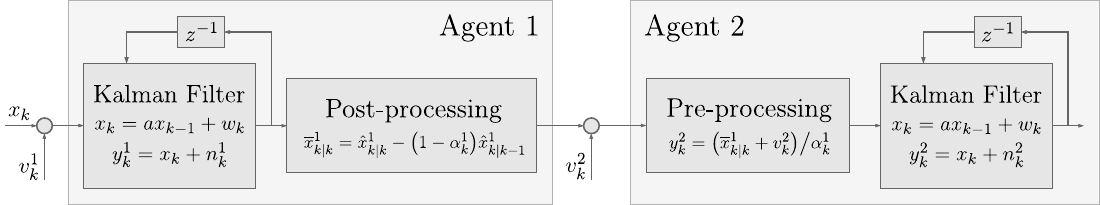}
    \caption{Interconnection of $m=2$ agents implementing a PP setup.}
    \label{fig:CASCADE ARCHITECTURE}
\end{figure*}

In this subsection, we derive recursive expressions for an agent $i\in\mathcal{I}$ embedded in a PP setup, as depicted in Figure~\ref{fig:CASCADE ARCHITECTURE} for the case $m=2$. In doing so, we reduce ourselves to work with the multi-agent system introduced in Section~\ref{subsec:MODEL SELECTION}.

The prediction step is independent of the newest collected measurement $\tilde{y}_k^{i}$, and hence
\begin{equation}\label{eq:PREDICTION STEP CASCADE, m AGENTS}
    \begin{aligned}
        & p_{k|k-1}^{i} = a^2 p_{k-1|k-1}^{i} + q, \\
        & \hat x_{k|k-1}^{i} = a \hat x_{k-1|k-1}^{i},
    \end{aligned}
\end{equation}
can be updated independently by each $i\in\mathcal{I}$ using {private prior} information of $x_{k-1}$, that is, $p^{i}_{k-1|k-1}$ and $\hat{x}^{i}_{k-1|k-1}$.

On the other hand, the a-posteriori update depends on the estimates of other decision-makers, and its derivation is more involved. 
To draw some intuition, we start from the case of $m=2$, and then provide general expressions for $m\geq 2$.
At each time step $k$, the leftmost agent can observe the unknown state directly through $y_k^{1} = x_k + v_k^{1}$. This observation model is the same as in \eqref{eq:OBSERVATION MODELS}, but with $n_k^{1} = v_k^{1}$. Thus, the posterior parameters for agent 1 are obtained from \eqref{eq:KALMAN UPDATE} as
\begin{equation}\label{eq:KALMAN UPDATE CASCADE AGENT 1}
    \begin{aligned}
        & \alpha_k^{1} = p_{k|k-1}^{1} / (p_{k|k-1}^{1} + s^{1}), \\
        & p_{k|k}^{1} = p_{k|k-1}^{1}(1 - \alpha_k^{1}), \\
        & \hat x_{k|k}^{1} = \hat x_{k|k-1}^{1} + \alpha_k^{1}(\tilde{y}_k^{1} - \hat x_{k|k-1}^{1}).
    \end{aligned}
\end{equation}
The estimate $\hat x_{k|k}^{1}$ is forwarded to the next agent only after undergoing a post-processing operation
\begin{equation}\label{eq:POST-PROCESSING}
    \overline x_{k|k}^{1} 
    \coloneqq \hat{x}_{k|k}^{1} - (1 - \alpha_k^{1}) \hat{x}_{k|k-1}^{1} 
    \overset{\eqref{eq:KALMAN UPDATE}}{=} \alpha_k^{1} ({{x}_k} + v_k^{1}),
\end{equation}
that removes private prior information available only to agent 1.
As illustrated in Figure~\ref{fig:CASCADE ARCHITECTURE}, agent 2 receives the signal $\overline x_{k|k}^{1} + v_k^{2}$. A pre-processing step yields
\begin{equation}\label{eq:PRE-PROCESSING}
    \begin{aligned}
        y_k^{2} 
        \coloneqq \big( \overline x_{k|k}^{1} + v_k^{2} \big) \big/ \alpha_k^{1}  
        \overset{\eqref{eq:POST-PROCESSING}}{=} x_k + v_k^{1} + v_k^{2} / \alpha_k^{1},
    \end{aligned}
\end{equation}
which is then handled by the Kalman filter of agent 2. By doing this, we obtain the same observation model as in \eqref{eq:OBSERVATION MODELS}, where $n_k^{2} \coloneqq v_k^{1} + v_k^{2} / \alpha_k^{1}$ is a zero-mean Gaussian distributed random variable with variance 
\begin{equation}\label{eq:EQUIVALENT NOISE, AGENT 2}
    r_k^{2} = s^{1} + s^{2} / \big( \alpha_k^{1} \big)^\bold{2}.
\end{equation}
This is a direct consequence of the fact that all the previously defined random variables are independent and Gaussian distributed.
Agent 2 can then retrieve its posterior estimate $\hat{x}_{k|k}^{2}$ by substituting \eqref{eq:EQUIVALENT NOISE, AGENT 2} into the expression of $\alpha_k^{2}$.

Let now $m\geq 2$, and consider an arbitrary $i\in\mathcal{I}\setminus \{1\}$. As before, one can modify \eqref{eq:POST-PROCESSING} and \eqref{eq:PRE-PROCESSING} to hold for an arbitrary index $i$ and combine them as

\begin{equation*}\label{eq:PRE and POST PROCESSING, m AGENTS}
    \begin{aligned}
         y_k^{i} 
        &\coloneqq \big[\hat x_{k|k}^{i-1} - (1 - \alpha_k^{i-1})\hat x_{k|k-1}^{i-1} + v_k^{i}\big] \big/ \alpha_k^{i-1} \\
        & = x_k + n_k^{i-1} + v_k^{i}/ \alpha_k^{i-1},
    \end{aligned}
\end{equation*}
for every $i\in\mathcal{I}\setminus \{1\}$. Using a recursive expression, we are able to define the equivalent noise $n_k^{i}\coloneqq n_k^{i-1} + v_k^{i}/\alpha_k^{i-1}$, which is a zero-mean Gaussian variable with variance
\begin{equation}\label{eq:EQUIVALENT NOISE, AGENT m}
    \begin{aligned}
        r_k^{i} 
        = s^{1} + \sum_{j=2}^{i} \frac{s^{j}}{ \big( \alpha_k^{j-1} \big)^\bold{2}}
        , \quad i\in\mathcal{I}.
    \end{aligned}
\end{equation}
\begin{remark}
    \textit{The variance $r_k^i$ is strictly increasing on $\mathcal{I}$, that is, $r_k^{i}>r_k^j$ for every $i,j\in\mathcal{I}$ such that $i>j$.}
\end{remark}
The posterior parameters $\hat{x}_{k|k}^{i}$ and ${p}_{k|k}^{i}$ are computed by substituting \eqref{eq:EQUIVALENT NOISE, AGENT m} into the measurement update in \eqref{eq:KALMAN UPDATE} to obtain
\begin{equation}\label{eq:CORRECTION STEP CASCADE, m AGENTS}
    \begin{aligned}
        & \alpha_k^{i} = p_{k|k-1}^{i}\Big/ \Big[p_{k|k-1}^{i} + s^{1} + \textstyle\sum_{j = 2}^{i} s^{j}\big/\big( \alpha_k^{j-1} \big)^\bold{2}\Big], \\
        & p_{k|k}^{i} = p_{k|k-1}^{i}(1 - \alpha_k^{i}), \\
        & \hat x_{k|k}^{i} = \hat x_{k|k-1}^{i} + \alpha_k^{i} (\tilde{y}_k^{i} - \hat x_{k|k-1}^{i} ),
        \qquad\qquad i\in\mathcal{I}.
    \end{aligned}
\end{equation}


\subsection{Word-of-Mouth Setup}\label{subsec:PUBLIC PRIOR}
In the WoM framework, the decision-makers are again interconnected in series but share a common prior. In particular, agent $m$ is enforcing her posterior to be the next prior for everyone else, as is illustrated in Figure~\ref{fig:WoM ARCHITECTURE} for the case $m=2$.

Our previous derivations can be used as the starting point for describing
WoM social learning. After substituting $x_{k-1|k-1}^{i} = x_{k-1|k-1}^m$ into \eqref{eq:PREDICTION STEP CASCADE, m AGENTS} we obtain
\begin{equation}\label{eq:PREDICTION STEP WoM, m AGENTS}
    \begin{aligned}
        & p_{k|k-1}^{i} = a^2 p_{k-1|k-1}^{m} + q, \\
        & \hat x_{k|k-1}^{i} = a \hat x_{k-1|k-1}^{m},
    \end{aligned}
    \qquad i\in\mathcal{I}.
\end{equation}
Perhaps not surprisingly, the $m$-th agent is also imposing her one-step-ahead prediction, that is, $\hat{x}^i_{k|k-1} = \hat{x}^m_{k|k-1}$ and ${p}^i_{k|k-1} = {p}^m_{k|k-1}$ for every $i$. The a-posteriori update 
\begin{equation}\label{eq:CORRECTION STEP 1 WoM, m AGENTS}
    \begin{aligned}
        & \alpha_k^{i} = p_{k|k-1}^{m}\Big/ \Big[p_{k|k-1}^{m} + s^{1} + \textstyle\sum_{j = 2}^{i} s^{j}\big/\big( \alpha_k^{j-1} \big)^\bold{2}\Big], \\
        & p_{k|k}^{i} = p_{k|k-1}^{m}(1 - \alpha_k^{i}), \\
        & \hat x_{k|k}^{i} = \hat x_{k|k-1}^{m} + \alpha_k^{i} (\tilde{y}_k^{i} - \hat x_{k|k-1}^{m} ),
    \end{aligned}
\end{equation}
yields the posterior parameters for WoM social learning of agent $i$. By defining $\gamma_k^{i} \coloneqq s^{i}\big/\big( \alpha_k^{i-1} \big)^\bold{2}$ as the new variance contribution to the equivalent noise entering immediately before $i\in\mathcal{I}\setminus \{1\}$, we rewrite the posterior variance as
\begin{equation}\label{eq:CORRECTION STEP 2 WoM, m AGENTS}
    \begin{aligned}
        & p_{k|k}^{i} 
        = \dfrac{p_{k|k-1}^{m} \big(s^1 + \textstyle\sum_{j=2}^{i} \gamma_k^{j} \big)}{p_{k|k-1}^{m} + s^{1} + \textstyle\sum_{j=2}^{i} \gamma_k^{j}},
        &i\in\mathcal{I}, \\
        & \gamma_k^{i}
        = s^{i} \left[ 1 + \dfrac{\big(s^1 + \textstyle\sum_{j=2}^{i-1} \gamma_k^j \big)}{p_{k|k-1}^{m}} \right]^{\boldsymbol{2}},&i \in \mathcal{I}\setminus \{1\}.
    \end{aligned}
\end{equation}

\begin{remark} \textit{Let $i\in\mathcal{I}\setminus\{1\}$. In the WoM setup, the random variables $y_k^{i}$ and $\hat{x}^{i-1}_{k|k} + v_k^{i}$ are observationally equivalent \cite{Vives01}, as there exists a deterministic bijection between them. Intuitively, they both express the new information available about $x_k$, given the measurements.}
\end{remark}

\begin{figure*}[t!]
    \vspace{0.4 em}
    \centering
    \includegraphics[width=\textwidth]{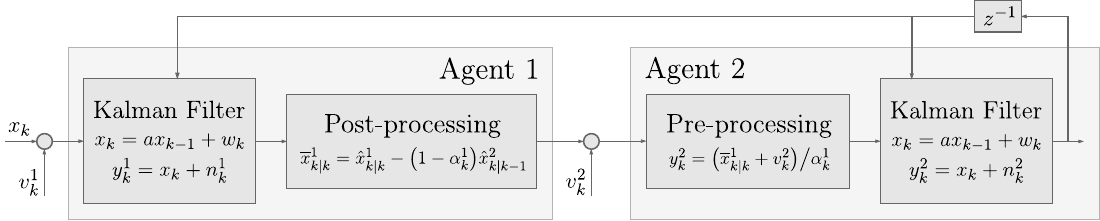}
    \caption{Interconnection of $m=2$ agents implementing a WoM setup.}
    \label{fig:WoM ARCHITECTURE}
\end{figure*}



%% file: theoretical_analysis.tex

\section{Analysis of PP and WoM Learning}\label{sec:THEORETICAL ANALYSIS}
In this section, we study the multi-agent estimation problems formulated in Sections~\ref{subsec:PRIVATE PRIOR} and \ref{subsec:PUBLIC PRIOR}. In particular, we are interested in characterizing the asymptotic behavior of the $m$-th agent's predictive and posterior variances for both scenarios. With this in mind, we first prove convergence of $p_{k|k-1}^{i}$ to a unique positive fixed point for all $i\in\mathcal{I}$ in the PP setup. The reader should notice that this implies convergence of the associated posterior $p_{k|k}^{i}$. After this, we show that a unique positive fixed point also exists for the WoM case \eqref{eq:PREDICTION STEP WoM, m AGENTS}.

\subsection{Private-Prior Setup}
In Section~\ref{subsec:PRIVATE PRIOR}, each agent is estimating the state of a time-invariant and asymptotically stable system \eqref{eq:AUTOREGRESSIVE DYNAMICS} while retaining their individual beliefs. To account for informational asymmetry, the agents adopt custom observation models \eqref{eq:OBSERVATION MODELS} that differ on their measurement noise parameters \eqref{eq:EQUIVALENT NOISE, AGENT m}. 
The first agent is the only one employing a time-invariant observation model: $n_k^{i}$ is stationary because $r_k^1 = s^{1}$ is constant.
Instead, every $i\in\mathcal{I}\setminus \{1\}$ is implementing a time-varying state-space model, since $n_k^{i}$ has time-dependent variance $r_k^{i}$. The following preparatory lemma provides sufficient conditions for convergence to a steady-state Kalman filter.


\begin{lemma}\label{lemma:CONVERGENCE TO STEADY STATE KALMAN FILTER}
\textit{Consider the iteration}
\begin{align*}
    p_{k+1|k} 
    = 
    a^{\boldsymbol{2}} p_{k|k-1} - \frac{(a p_{k|k-1})^{\boldsymbol{2}}}{ p_{k|k-1} + r_{k}} + q,
\end{align*}
\textit{where $a\in(-1,1)$, $q \in \mathbb{R}_{>0}$, $(r_k)$ is a sequence in $\mathbb{R}_{>0}$ such that $r_k \to {r_\infty} \in \mathbb{R}_{>0}$, and $p_{1|0} \in \mathbb{R}_{>0}$ is arbitrary. Then, $p_{k|k-1}$ converges to the unique positive solution $p_\infty$ of the discrete algebraic Riccati equation (DARE)}
\begin{equation*}\label{eq:DARE original}
p_\infty = a^{\boldsymbol{2}} p_\infty -\dfrac{(a p_\infty)^{\boldsymbol{2}}}{p_\infty + r_\infty} + q.
\end{equation*}
\textit{Also, $(1-\alpha_\infty)\, a \in (-1, 1)$, where $\alpha_\infty = p_\infty / (p_\infty + {r_\infty})$.}
\end{lemma}

\begin{proof}
The properties of the DARE are well known~\cite[Chapter 5]{anderson2012optimal}. The update rule for $p_{k+1|k}$ can be written as
\begin{align*}
\begin{bmatrix}
p_{k+1|k} \\
1
\end{bmatrix}
= \beta_{k|k-1} \begin{bmatrix}
a^{\boldsymbol{2}} r_k + q & qr_k \\
1 & r_k
\end{bmatrix}
\begin{bmatrix}
p_{k|k-1} \\
1
\end{bmatrix},
\end{align*}
where $\beta_{k|k-1} := 1/(p_{k|k-1} + r_k)$ is a normalization factor. These iterations correspond to the power method (with a non-standard normalization factor), whose convergence is assured since the square matrix above is strictly positive (so it has a unique positive eigenvalue of maximum absolute value, by Perron's theorem~\cite{Horn-Johnson-12}) and $r_k \to r_\infty$ as $k \to \infty$; c.f.~\cite{Hardt-Price-14}. Also, the limit value of $p_{k+1 | k}$ has to satisfy the DARE.
\end{proof} 

The following cascade behavior is expected to occur: after the Kalman filter of agent $i$  converges to steady-state, so will do the Kalman filter of agent $i+1$, for every $i \in \mathcal{I}\setminus\{m\}$. This result is formalized in the following theorem:

\begin{theorem}\label{th:convergence of the cascade}
    \textit{The PP architecture implementing \eqref{eq:PREDICTION STEP CASCADE, m AGENTS} and \eqref{eq:CORRECTION STEP CASCADE, m AGENTS} admits a unique positive solution $p_\infty^{i}$ to}
    \begin{equation*}
        \begin{aligned}
            p_\infty^{i} = a^{\boldsymbol{2}}  p_\infty^{i} - \dfrac{(a p_\infty^{i})^{\boldsymbol{2}}}{p_\infty^{i} + r_\infty^{i}} + q, \qquad i\in\mathcal{I},
        \end{aligned} 
    \end{equation*}    
    \textit{where $r_\infty^i = s^1 + \sum_{j=2}^{i}s^{j}/ (\alpha_\infty^{j-1})^{\boldsymbol{2}}$. 
    In addition, $p_{k|k-1}^{i}\rightarrow p_\infty^{i}$ for any $p_{1|0}^{i}\in\mathbb{R}_{>0}$, and $(1 - \alpha_\infty^{i}) a \in (-1,1)$. 
    }
\end{theorem} 

\begin{proof} We proceed by induction on $i\in\mathcal{I}$, and prove that each of the agents' estimators does converge to a unique stationary Kalman filter. 
By \cite[Chapter 4]{anderson2012optimal}, $p_{k|k-1}^1\rightarrow p_\infty^1$ as $k\to\infty$ for any initial condition $p_{1|0}^1\in\mathbb{R}_{>0}$, where $p_\infty^1$ is the unique positive solution of
\begin{equation*}
    p_\infty^1 = a^{\boldsymbol{2}} p_\infty^1 - (ap_\infty^1)^{\boldsymbol{2}} / (p_\infty^{1} + s^1) + q,
\end{equation*}
and $(1 - \alpha_\infty^1) \, a \in (-1,1)$. Consider $i=2$. The equivalent noise $(n_k^2)$ has variance that evolves according to the sequence $(r_k^2) = (s^1 + s^2/(\alpha_k^1)^{\boldsymbol{2}})$, from \eqref{eq:EQUIVALENT NOISE, AGENT m}. Note that $(r_k^2)$ is a sequence in $\mathbb{R}_{>0}$ and that $r_k^2 \rightarrow r_\infty^2\coloneqq s^1 + s^2/(\alpha_\infty^1)^{\boldsymbol{2}} > 0$ when $k\to\infty$. Therefore,
by
%
%
Lemma~\ref{lemma:CONVERGENCE TO STEADY STATE KALMAN FILTER}, 
%
$p_{k|k-1}^2$ converges to the unique positive solution $p_\infty^2\geq q>0$ of
\begin{equation*}
    \begin{aligned}
        & p_\infty^{2} = a^{\boldsymbol{2}}  p_\infty^{2} - (a p_\infty^{2})^{\boldsymbol{2}}/(p_\infty^{2} + r_\infty^{2}) + q,
    \end{aligned}
\end{equation*}
for any initial condition $p_{1|0}^{2}\in\mathbb{R}_{>0}$, and $(1 - \alpha_\infty^2) \, a \in(-1,1)$. Agent 2 also converges to a steady-state Kalman filter. To prove the statement for any $i\in\mathcal{I}\setminus \{1\}$, assume that each of her predecessors $j\in\{1,\ldots,i-1\}$ converges to a steady-state Kalman filter. Therefore $n_k^{i}$ converges in distribution to a Gaussian random variable $n_\infty^{i}$ with second-order moment $r_\infty^{i} \coloneqq s^1 + \sum_{j=2}^{i}s^{j}/(\alpha_\infty^{j-1})^{\boldsymbol{2}}$. Following the same rationale as before, Lemma~\ref{lemma:CONVERGENCE TO STEADY STATE KALMAN FILTER} yields $p_{k|k-1}^{i}\rightarrow p_\infty^{i}$ as $k\to\infty$ for every initial condition $p_{1|0}^{i}\in\mathbb{R}_{>0}$, where $p_\infty^{i}$ is the unique positive solution of
\begin{equation*}
    \begin{aligned}
        p_\infty^{i} = a^{\boldsymbol{2}}  p_\infty^{i} - (a p_\infty^{i})^{\boldsymbol{2}}/(p_\infty^{i} + r_\infty^{i}) + q,
    \end{aligned}
\end{equation*}
and $(1 - \alpha_\infty^i) a \in (-1,1)$, thus proving the statement.
\end{proof}

\subsection{Word-of-Mouth Setup}
We characterize the asymptotic behavior of the predictive variance under the WoM setup by studying
\begin{equation}\label{eq:MAP BETWEEN PREDICTED VARIANCES}
    T(p_{k|k-1}^{m})
    = 
    a^{\boldsymbol{2}}\left(\dfrac{p_{k|k-1}^{m} r_k^{m}}{p_{k|k-1}^{m} + r_k^{m}}\right) + q,
\end{equation}
obtained by substituting \eqref{eq:CORRECTION STEP 1 WoM, m AGENTS} into \eqref{eq:PREDICTION STEP WoM, m AGENTS} for $i=m$.
The following lemma is needed for the analysis.
\begin{lemma}\label{lemma:VARIANCE AS POSITIVE POLYNOMIAL}
    \textit{In the WoM setup, $r_k^{i}$ is a positive and strictly decreasing function of $p_{k|k-1}^{m}\in\mathbb{R}_{>0}$ of the form} 
    \begin{equation}\label{eq:VARIANCE AS POSITIVE POLYNOMIAL}
        f(p_{k|k-1}^{m},i)
        =
        \sum_{j = 0}^{2^{i}-2} \dfrac{c_j}{(p_{k|k-1}^{m})^{\boldsymbol{j}}} \qquad i\in\mathcal{I}\setminus \{1\},
    \end{equation}
    \textit{with coefficients $c_j\in\mathbb{R}_{>0}$, for all $j\in\{0,\ldots,2^{\boldsymbol{i}}-2\}$.}
\end{lemma}
\begin{proof}
    It has been proved in \cite{slowConvergenceSL} that $\gamma_k^{i}$ is a polynomial of order $2^{\boldsymbol{i}}-2$ in $1/p_{k|k-1}^{m}$ having positive coefficients. Since \eqref{eq:EQUIVALENT NOISE, AGENT m} can be recast as $r_k^{i} = s^1 + \gamma_k^2 + \cdots + \gamma_k^i$, it is also a polynomial of degree $2^{\boldsymbol{i}}-2$ in $1/p_{k|k-1}^m$ with coefficients $c_j\in\mathbb{R}_{>0}$ for every $j\in\{0, \ldots, 2^{\boldsymbol{i}}-2\}$, thus yielding \eqref{eq:VARIANCE AS POSITIVE POLYNOMIAL}. Because $p_{k|k-1}^{m}\in\mathbb{R}_{>0}$, $f$ is the sum of one positive constant (i.e., $c_0$) and $2^{\boldsymbol{i}}-2$ positive, monotonically decreasing functions of $p_{k|k-1}^m$, proving the statement.
\end{proof}

We can now state one of the main results of this paper, i.e., the existence of a unique positive fixed point for \eqref{eq:MAP BETWEEN PREDICTED VARIANCES}.

\begin{theorem}\label{th:xistence of a unique fixed point in WoM}
    \textit{Consider the WoM updates in \eqref{eq:PREDICTION STEP WoM, m AGENTS} and \eqref{eq:CORRECTION STEP 2 WoM, m AGENTS}. Then, Equation \eqref{eq:MAP BETWEEN PREDICTED VARIANCES} has a unique positive fixed point $p_\infty^m$.}

\end{theorem}
\begin{proof}
    Every fixed point $p_\infty^m$ for \eqref{eq:MAP BETWEEN PREDICTED VARIANCES} must satisfy
    \begin{equation}\label{eq:FIXED POINT SECOND EQUATION}
        f(p_\infty^m, m) = \dfrac{p_\infty^m (q - p_\infty^m)}{p_\infty^m(1 - a^{\boldsymbol{2}}) - q}.
    \end{equation}
    To simplify the analysis, we replace $p_\infty^m$ by $p\in\mathbb{R}_{>0}$, and study both sides of \eqref{eq:FIXED POINT SECOND EQUATION}. The right-hand side, denoted as
    \begin{equation*}
        g(p) \coloneqq \dfrac{p(q-p)}{p(1 - a^{\boldsymbol{2}}) - q},
    \end{equation*}
    is positive if $p\in(q,q/(1-a^{\boldsymbol{2}}))$, and non-positive otherwise. In addition, $g(p)\rightarrow +\infty$ when $p\rightarrow q/(1-a^{\boldsymbol{2}})$ from the left, and is monotonically increasing when 
    \begin{equation*}
        p 
        \in 
        \left( 
        \dfrac{q(1 - |a|)}{1-a^{\boldsymbol{2}}},   
        \dfrac{q}{1-a^{\boldsymbol{2}}}\right) 
        \cup 
        \left( 
        \dfrac{q}{1-a^{\boldsymbol{2}}},
        \dfrac{q(1 + |a|)}{1-a^{\boldsymbol{2}}}
        \right)
    \end{equation*}
    for every $q\in\mathbb{R}_{>0}$ and $a\in(-1,1)$. From Lemma~\ref{lemma:VARIANCE AS POSITIVE POLYNOMIAL}, $f$ is positive and decreasing in $p$. Therefore, $f$ and $g$ must intersect at exactly one point $p_\infty\in\big(q,q/(1-a^{\boldsymbol{2}})\big)$, proving the existence of a unique positive fixed point.
\end{proof}

Therefore, if the positive sequence $(p_{k|k-1}^{m})$ is convergent, then it must converge to $p_\infty^{m}$.
The following theorem establishes the convergence of $(p_{k|k-1}^{m})$ to $p_\infty^{m}$ for the case $m=2$.
\begin{theorem}\label{th:convergence in WoM, 2 agents}
    \textit{Let $m=2$ and $p_\infty^m$ be the unique positive fixed point of \eqref{eq:MAP BETWEEN PREDICTED VARIANCES}. Then $p_{k|k-1}^m\rightarrow p_\infty^m$ for any $p_{1|0}^m\in\mathbb{R}_{>0}$.}
\end{theorem}
\begin{proof} 
    We want to show that \eqref{eq:MAP BETWEEN PREDICTED VARIANCES} is a contraction on $\mathbb{R}_{>0}$ for $m=2$.
    For the sake of readability, we are going to suppress the second argument in \eqref{eq:VARIANCE AS POSITIVE POLYNOMIAL}.
    We start by finding an upper bound for
    \begin{equation}\label{eq:CONTRACTION CHECK}
        \begin{aligned}
            T(p') - T(p'') 
            \overset{\eqref{eq:MAP BETWEEN PREDICTED VARIANCES}}{=}
            a^{\boldsymbol{2}}
            \left( 
            \dfrac{p'f(p')}{p' + f(p')}
            -
            \dfrac{p''f(p'')}{p'' + f(p'')}
            \right) \\
            = 
            a^{\boldsymbol{2}}
            \left(p' - p'' + \dfrac{(p'')^{\boldsymbol{2}}}{p'' + f(p'')}
            -
            \dfrac{(p')^{\boldsymbol{2}}}{p' + f(p')}
            \right),
        \end{aligned}
    \end{equation}
    where we assume $p' \geq p'' > 0$ without loss of generality. Lemma~\ref{lemma:VARIANCE AS POSITIVE POLYNOMIAL} then implies that $f(p'') \geq f(p') > 0$.
    One can see that the inequalities $(p'')^{\boldsymbol{2}}p' \leq p''(p')^{\boldsymbol{2}}$ and $(p'')^{\boldsymbol{2}}f(p') \leq 
    f(p'')(p')^{\boldsymbol{2}}$ hold simultaneously, and therefore
    \begin{equation*}\label{eq:SIGN CHECK 1}
            \dfrac{(p'')^{\boldsymbol{2}}}{p'' + f(p'')}
            -
            \dfrac{(p')^{\boldsymbol{2}}}{p' + f(p')}
            \leq 0.
    \end{equation*}
    Substituting the latter in \eqref{eq:CONTRACTION CHECK} yields the upper bound
    \begin{equation*}
        T(p') - T(p'') \leq a^{\boldsymbol{2}}(p'-p'').
    \end{equation*}
    To obtain a lower bound on \eqref{eq:CONTRACTION CHECK} we now show that
    \begin{equation}\label{eq:CONTRACTION CHECK 2}
        \begin{aligned}
            \frac{1}{a^{\boldsymbol{2}}} \dfrac{\mathrm{d}T(p)}{\mathrm{d}p} 
            \overset{\eqref{eq:MAP BETWEEN PREDICTED VARIANCES}}{=}
            2f^{\boldsymbol{2}}(p) + 2pf(p) + p^{\boldsymbol{2}} \left( \dfrac{\mathrm{d}f(p)}{\mathrm{d}p} +1 \right)
            \geq -1
        \end{aligned}
    \end{equation}
    for every $p\in\mathbb{R}_{>0}$. From Lemma~\ref{lemma:VARIANCE AS POSITIVE POLYNOMIAL}, when $i=2$:
    \begin{equation*}\label{eq:FUNCTION f AND DERIVATIVE, 2 AGENTS}
        f(p) 
        = 
        c_0 + \dfrac{c_1}{p} + \dfrac{c_2}{p^{\boldsymbol{2}}}
        \implies
        \dfrac{\mathrm{d}f(p)}{\mathrm{d}p}
        = 
        -\dfrac{c_1}{p^{\boldsymbol{2}}} -\dfrac{2c_2}{p^{\boldsymbol{3}}},
    \end{equation*}
    where $c_0,c_1,c_2\in\mathbb{R}_{>0}$. Substituting inside \eqref{eq:CONTRACTION CHECK 2} yields
    \begin{multline*}
        \frac{1}{a^{\boldsymbol{2}}} \dfrac{\mathrm{d}T(p)}{\mathrm{d}p}
        = 
        p^{\boldsymbol{2}} + 2c_0p + (2c_0^{\boldsymbol{2}} + c_1) +
        \dfrac{4c_0c_1}{p} \\
        + \dfrac{2c_1^{\boldsymbol{2}} + 4c_0c_2}{p^{\boldsymbol{2}}} +
        \dfrac{4c_1c_2}{p^{\boldsymbol{3}}} +
        \dfrac{2c_2^{\boldsymbol{2}}}{p^{\boldsymbol{4}}} 
        \geq -1, 
        \quad \forall p\in\mathbb{R}_{>0}.
    \end{multline*}
    Therefore, $|T(p') - T(p'')|\leq a^{\boldsymbol{2}} |p' - p''|$, and hence $T$ is a contraction on $\mathbb{R}_{>0}$.
    Our statement then follows directly from the Banach fixed point theorem ~\cite[Chapter 2]{kolmogorov1975introductory}.
\end{proof}


%% file: numerical_results.tex

\section{Numerical Examples$^1$}\label{sec:EXPERIMENTS}
This section presents numerical examples that corroborate the theoretical results of Section~\ref{sec:THEORETICAL ANALYSIS}. In particular, the estimation performance of agents operating under the PP and WoM setups is compared. In addition, we numerically verify the convergence of the prediction error to a stationary process.

\subsection{Experimental Setting}
We consider the multi-agent setups of Sections~\ref{subsec:PRIVATE PRIOR} and \ref{subsec:PUBLIC PRIOR}, where $m=3$ agents try to estimate the scalar state of \eqref{eq:AUTOREGRESSIVE DYNAMICS} having initial condition $x_0\sim\mathcal{N}(\hat{x}_0,p_0)$, with $\hat{x}_0 = 25$ and $p_0 = 3$. The following system parameters will be considered: $a = 0.95$, $q = 1$, and $s^{i}=1$ for all $i\in\mathcal{I}=\{1,2,3\}$.
We initialize the Kalman filters of every agent $i\in\mathcal{I}$ to $\hat{x}_{0|0}^{i} = \hat x_0$ and $p_{0|0}^{i} = p_0$, which are equivalent to $\hat{x}_{1|0}^{i} = a\hat{x}_0 = 23.75$ and $p_{1|0}^{i} = a^{\boldsymbol{2}}p_0 + q \approx 3.7$, respectively. 
\let\thefootnote\relax
\hypersetup{hidelinks}
\footnote{
\noindent$^1$Code available at 
\nolinkurl{github.com/andreadacol98/on_wom_and_pp_sequential_sl}.
}

\subsection{Experimental Results}
A first batch of simulations addresses the asymptotic behavior of the variance for the one-step-ahead predictions, in both setups. Figure~\ref{fig:convergence of prediction variance} helps visualize the key aspects of PP and WoM learning, postulated in Section~\ref{sec:THEORETICAL ANALYSIS}. In a PP interconnection, convergence to a unique positive fixed point $p_\infty^{i}$ is attained by agent $i$ for all positive initial conditions. Note that different agents will have different fixed points. By Theorem~\ref{th:convergence of the cascade}, $p_\infty^{i}$ is an increasing function of $i$: one agent cannot predict the next state more accurately than its predecessors, that is, $0<p_\infty^1<p_\infty^2<p_\infty^3$. 
In the WoM setup, all the agents share the same fixed point for the prediction error variance, i.e., $p_\infty^1=p_\infty^2=p_\infty^3$, which is expected as the $m$-th agent is enforcing her belief. Furthermore, the simulations suggest convergence to this limit point for any initial condition, independently of the number of agents. Numerical values of the fixed points are summarized in Table~\ref{table:Prediction error variance}, and confirm the theoretical findings. One key fact emerges from our example: WoM improves the variance in the state-predictions of the last agent, at least asymptotically, i.e., $p_\infty^3$ is smaller in the WoM setup than in PP. {Section~\ref{subsec:discussion} provides an intuitive explanation for this.}

\begin{figure}[t!]
    \vspace{0.4 em}
    \centering
    \begin{subfigure}{\linewidth}
        \centering
        \includegraphics[width=\linewidth]{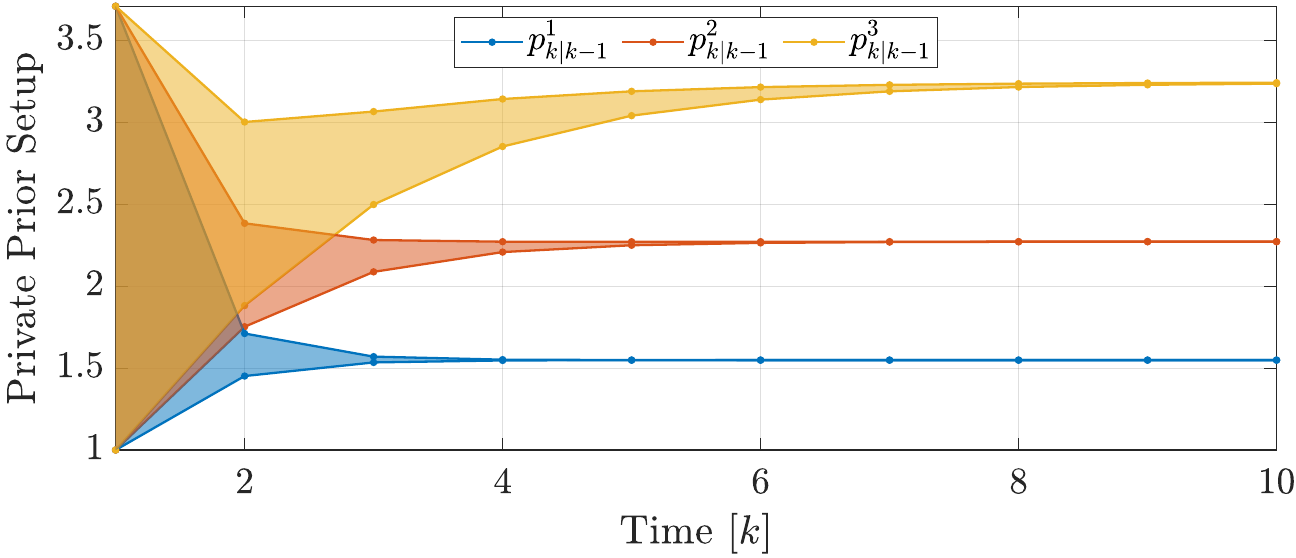}
        \label{fig: PREDICTED VARIANCE, CASCADE} 
    \end{subfigure}
    \vspace{-1.45 cm}
    
    \begin{subfigure}{\linewidth}
        \centering
        \includegraphics[width=\linewidth]{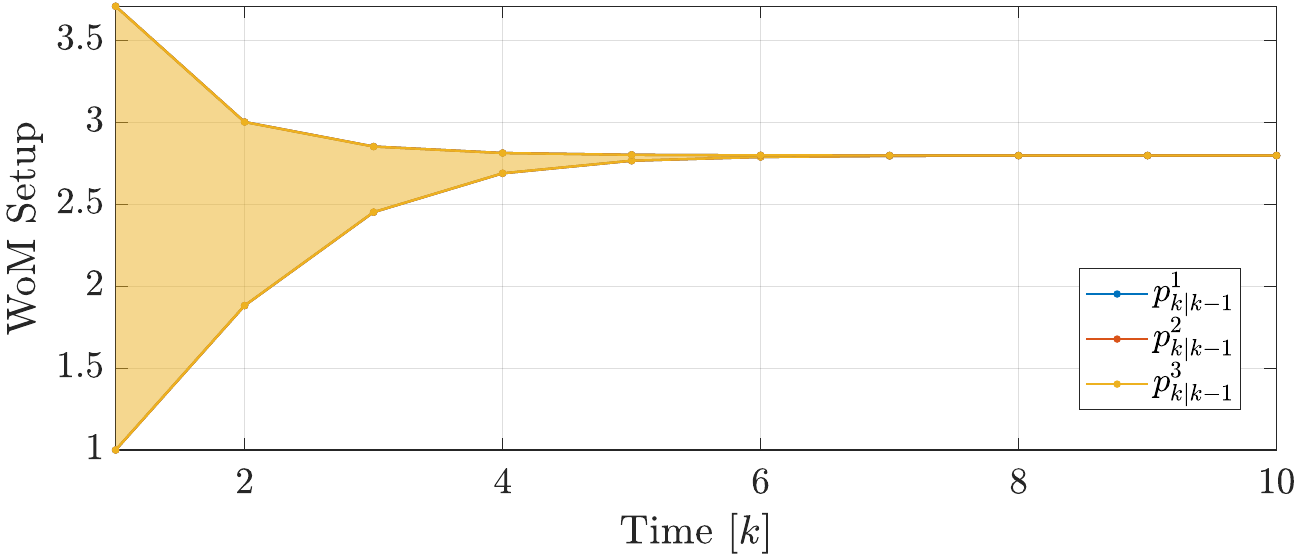}
        \label{fig: PREDICTED VARIANCE, WoM} 
    \end{subfigure}
    \vspace{-0.8 cm}
    \caption{Asymptotic behavior of the prediction error variances in the PP (upper panel) and WoM cases (lower panel).}
    \label{fig:convergence of prediction variance}
\end{figure}

Similar conclusions can be drawn for $(\alpha_{k}^{i})$ and $(p_{k|k}^{i})$ of different agents. Convergence of these sequences to unique positive fixed points follows directly from our previous considerations on $(p_{k|k-1}^{i})$, and is illustrated in Figure~\ref{fig:convergence of posterior variance}.
Interestingly, the agents converge to larger Kalman gains in the WoM setup, suggesting that more attention is put on other agents' actions than on the public prior.

\begin{figure}[ht]
    \vspace{0.4 em}
  \centering
  \begin{subfigure}{0.48\columnwidth}
    \includegraphics[width=\linewidth]{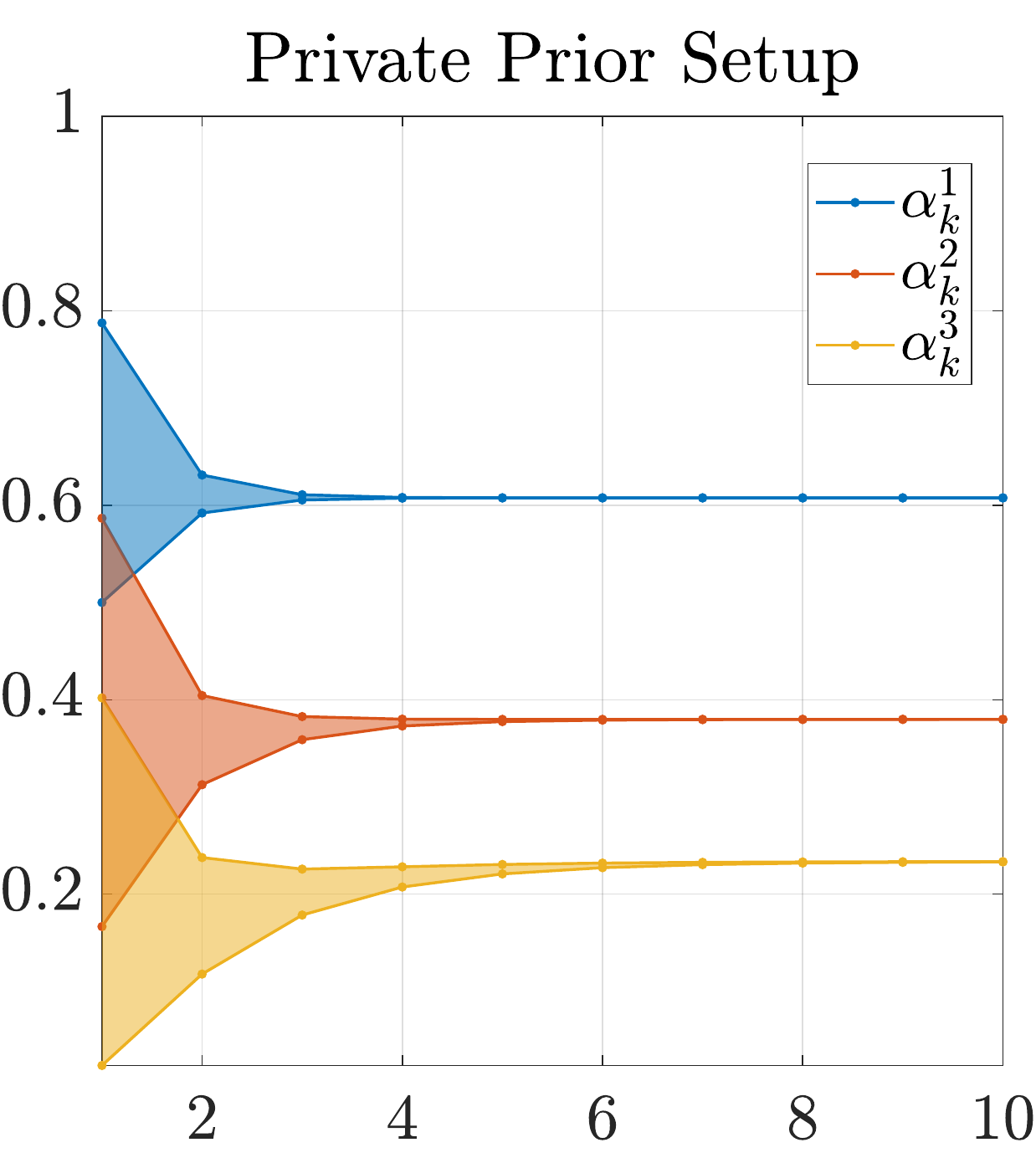}
  \end{subfigure}
  \hfill
  \begin{subfigure}{0.48\columnwidth}
    \includegraphics[width=\linewidth]{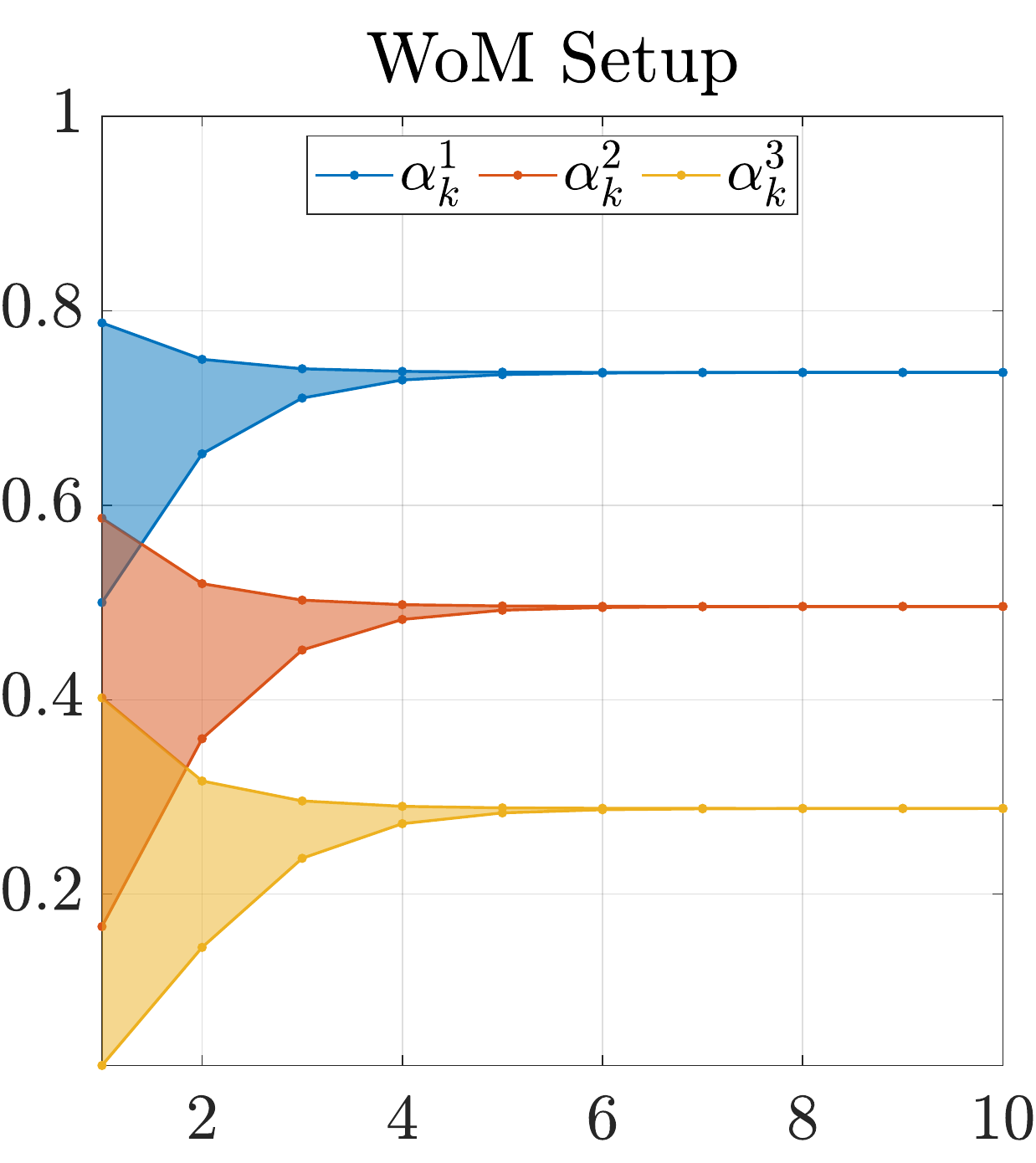}
  \end{subfigure}
  \begin{subfigure}{0.48\columnwidth}
    \includegraphics[width=\linewidth]{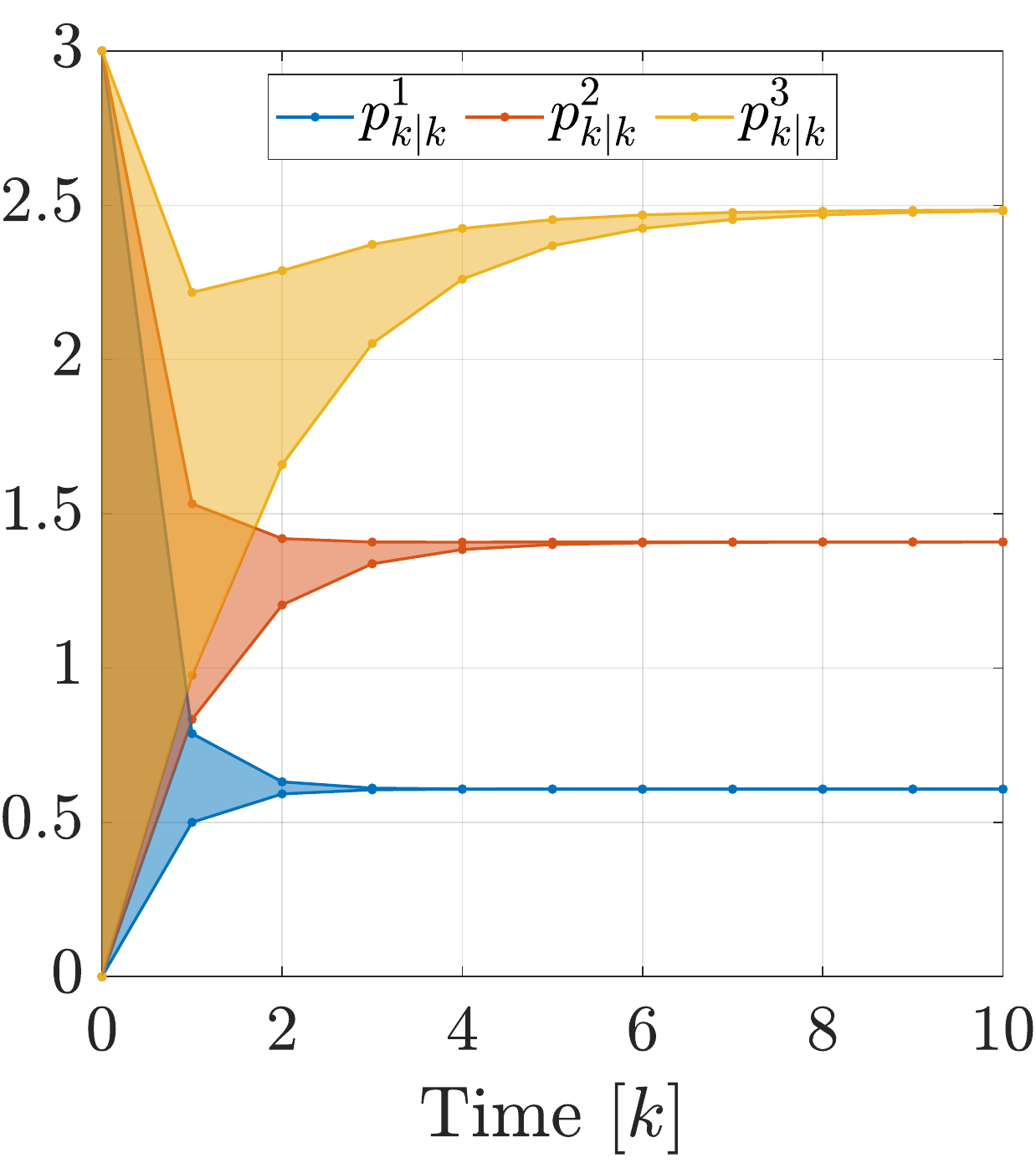}
  \end{subfigure}
  \hfill
  \begin{subfigure}{0.48\columnwidth}
    \includegraphics[width=\linewidth]{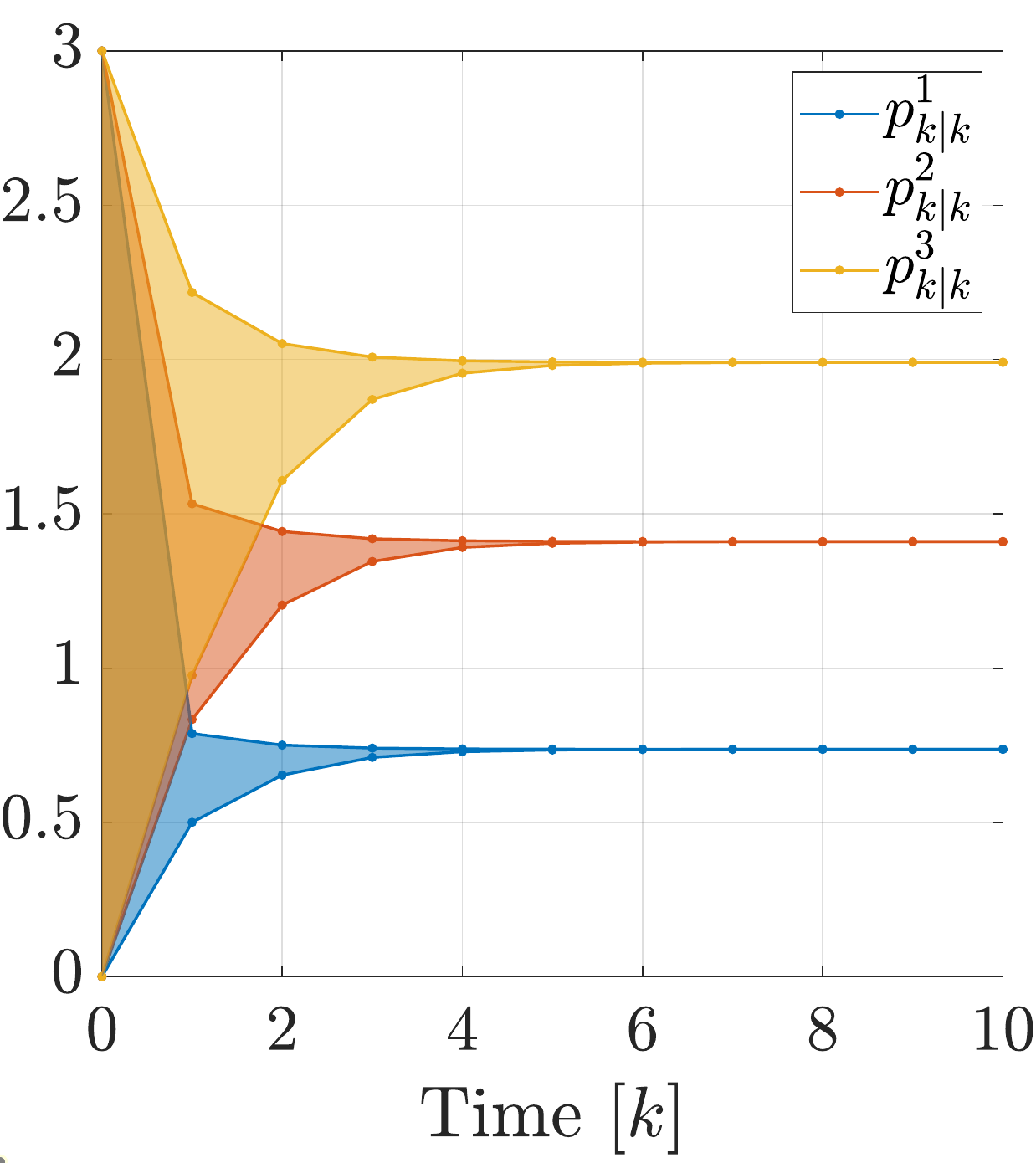}
  \end{subfigure}
  \caption{Asymptotic behavior of the Kalman gains (upper panels) and the posterior variances (lower panels) in the PP (left panels) and WoM cases (right panels).}
  \label{fig:convergence of posterior variance}
\end{figure}

To further support our claim of WoM being effective in reducing the variance of the estimation error for the last agent, we study the mean squared error of $\hat{x}_{k|k-1}^{i}$ and $\hat{x}_{k|k}^{i}$ for every agent $i$ numerically. Two separate experiments are run where we evaluate
\stepcounter{equation}
\begin{align}
        & \textrm{MSE}\big(\hat{x}_{k|k-1}^{i}\big) = \dfrac{1}{K}\sum_{k=1}^{K} \big( \hat{x}_{k|k-1}^{i} - \Tilde{x}_k \big)^{\boldsymbol{2}}, & i\in\mathcal{I}, \tag{\theequation.a} \label{eq:MSE OF PREDICTIONS} \\
        & \textrm{MSE}\big(\hat{x}_{k|k}^{i}\big) = \dfrac{1}{K}\sum_{k=1}^{K} \big( \hat{x}_{k|k}^{i} - \Tilde{x}_k \big)^{\boldsymbol{2}}, & i\in\mathcal{I}, \tag{\theequation.b} \label{eq:MSE OF POSTERIORS}
\end{align}
using the same sequence of realizations $(\tilde{x}_{k})$, with $K = 10^3$. 
Notice that, for large $K$, \eqref{eq:MSE OF PREDICTIONS} and \eqref{eq:MSE OF POSTERIORS} well approximate $\mathbb{E} \{ ( \hat{x}_{k|k-1}^{i} - x_k )^{\boldsymbol{2}} \}$ and $\mathbb{E} \{ ( \hat{x}_{k|k}^{i} - x_k )^{\boldsymbol{2}} \}$, respectively.
Numerical results from \eqref{eq:MSE OF PREDICTIONS} are found in Table~\ref{table:Prediction mean MSE}, and those from \eqref{eq:MSE OF POSTERIORS} are in Table~\ref{table:Posterior mean MSE}. Once again, the last agent benefits from the use of WoM, as reflected by the smaller values of the mean squared errors, compared to PP.

As one last example, we compare $(\tilde{x}_k^3)$ with the sequence of one-step-ahead predictions $(\hat{x}_{k|k-1}^3)$ in both setups. This time, we only address $i=3$, as similar conclusions can be drawn for the other agents. Figure~\ref{fig:convergence of prediction mean} shows how $(\tilde{x}_k^{3})$ stays inside the $3 \sigma$ confidence interval and $(\hat x_{k|k-1}^{3})$ converges to a zero-mean random process with variance $p_\infty^{3}$, for both setups.

\begin{table}[ht]
    \centering
    \caption{Fixed points for the prediction error variance.}
    \begin{tabular}{p{0.8cm}|p{0.8cm}|p{0.8cm}|p{0.8cm}|p{0.8cm}|p{0.8cm}}
        \hline
        \multicolumn{6}{c}{\textbf{Fixed points$-$prediction error variance}} \\
        \hline
        \multicolumn{2}{c|}{\textbf{Agent 1}} & \multicolumn{2}{c|}{\textbf{Agent 2}} & \multicolumn{2}{c}{\textbf{Agent 3}} \\
        \hline
        \multicolumn{1}{c|}{PP} & \multicolumn{1}{c|}{WoM} & \multicolumn{1}{c|}{PP} & \multicolumn{1}{c|}{WoM} & \multicolumn{1}{c|}{PP} & \multicolumn{1}{c}{WoM} \\
        \hline
        \multicolumn{1}{c|}{1.55} & \multicolumn{1}{c|}{2.79} & \multicolumn{1}{c|}{2.27} & \multicolumn{1}{c|}{2.79} & \multicolumn{1}{c|}{3.34} & \multicolumn{1}{c}{2.79} \\
        \hline 
    \end{tabular}
    \label{table:Prediction error variance}
\end{table}

\begin{table}[!ht]
    \centering
    \caption{MSE of the one-step-ahead predictions.}
    \begin{tabular}{p{0.8cm}|p{0.8cm}|p{0.8cm}|p{0.8cm}|p{0.8cm}|p{0.8cm}}
        \hline
        \multicolumn{6}{c}{\textbf{MSE$-$one-step-ahead prediction}} \\
        \hline
        \multicolumn{2}{c|}{\textbf{Agent 1}} & \multicolumn{2}{c|}{\textbf{Agent 2}} & \multicolumn{2}{c}{\textbf{Agent 3}} \\
        \hline
        \multicolumn{1}{c|}{PP} & \multicolumn{1}{c|}{WoM} & \multicolumn{1}{c|}{PP} & \multicolumn{1}{c|}{WoM} & \multicolumn{1}{c|}{PP} & \multicolumn{1}{c}{WoM} \\
        \hline
        \multicolumn{1}{c|}{1.54} & \multicolumn{1}{c|}{2.71} & \multicolumn{1}{c|}{2.13} & \multicolumn{1}{c|}{2.71} & \multicolumn{1}{c|}{3.24} & \multicolumn{1}{c}{2.71} \\
        \hline 
    \end{tabular}
    \label{table:Prediction mean MSE}
\end{table}

\begin{table}[!ht]
    \centering
    \caption{MSE of the a-posteriori estimates.}
    \begin{tabular}{p{1cm}|p{1cm}|p{1cm}|p{1cm}|p{1cm}|p{1cm}}
        \hline
        \multicolumn{6}{c}{\textbf{MSE$-$posterior estimate}} \\
        \hline
        \multicolumn{2}{c|}{\textbf{Agent 1}} & \multicolumn{2}{c|}{\textbf{Agent 2}} & \multicolumn{2}{c}{\textbf{Agent 3}} \\
        \hline
        \multicolumn{1}{c|}{PP} & \multicolumn{1}{c|}{WoM} & \multicolumn{1}{c|}{PP} & \multicolumn{1}{c|}{WoM} & \multicolumn{1}{c|}{PP} & \multicolumn{1}{c}{WoM} \\
        \hline
        \multicolumn{1}{c|}{0.66} & \multicolumn{1}{c|}{0.78} & \multicolumn{1}{c|}{1.58} & \multicolumn{1}{c|}{1.49} & \multicolumn{1}{c|}{2.73} & \multicolumn{1}{c}{2.10} \\
        \hline 
    \end{tabular}
    \label{table:Posterior mean MSE}
\end{table}


\subsection{Discussion} \label{subsec:discussion}

In this subsection we provide an intuitive explanation for why agents with large $i \in \mathcal{I}$ attain a smaller variance $p_{k|k}^{i}$ for the estimation error under the WoM setup, compared to PP. Intuitively, as $i$ increases, the prediction and posterior variances increase, due to the larger variance of the equivalent noise. Thus, in the WoM setup, all the agents (except for the last one) are forced to use a more spread-out prior than their own. Inevitably, agent 1, who alone can enjoy measurement noise with constant variance, will suffer a performance degradation, and output a larger Kalman gain compared to the (optimal) private-prior case. In other words, it will put more trust on the measurements than on its prior. As a result of this, the next agent will enjoy an equivalent noise with smaller variance. Together with being fed worse prior knowledge, she will also trust her measurements 
more than her prior belief, resulting in a larger Kalman gain. One can apply this reasoning inductively to all agents with $i > 1$, and when it comes to the last agent, the variance of her equivalent noise has decreased considerably compared to the PP case. This, plus the fact that she has been using her own prior (which corresponds to the public prior in the WoM case), makes the Kalman filter of agent $m$ converge to a steady-state filter with a smaller prediction error variance $p_\infty^m$ than in the PP setup.

It is important to note, however, that the data processing inequality (DPI)~\cite{Cover-Thomas-06} implies that if information is passed through a sequence of transformations without new inputs, the total information content (e.g., mutual information with the true state) can only stay the same or decrease. Therefore, even though WoM may yield a lower variance estimate for the final agent compared to PP in finite-agent settings, DPI implies that in sufficiently deep hierarchies, both frameworks eventually suffer from severe information degradation, fundamentally limiting the accuracy of the final estimate.

\begin{figure}[t!]
    \vspace{0.4 em}
    \centering
    \begin{subfigure}{\linewidth}
        \centering
        \includegraphics[width=\linewidth]{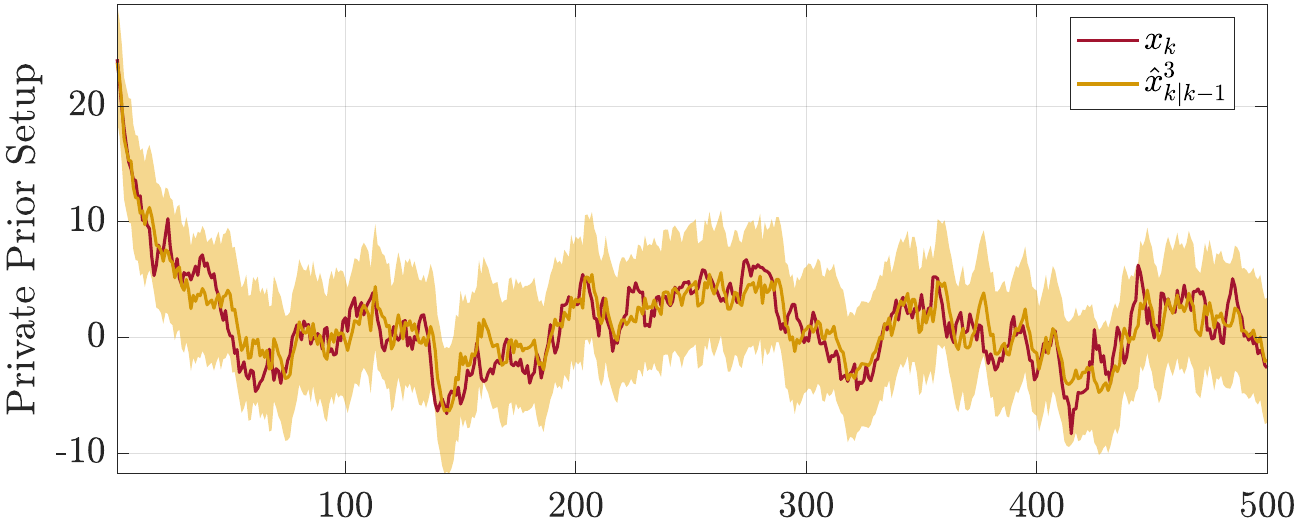} 
    \end{subfigure}
    \vspace{-0.7 cm}
    
    \begin{subfigure}{\linewidth}
        \centering
        \includegraphics[width=\linewidth]{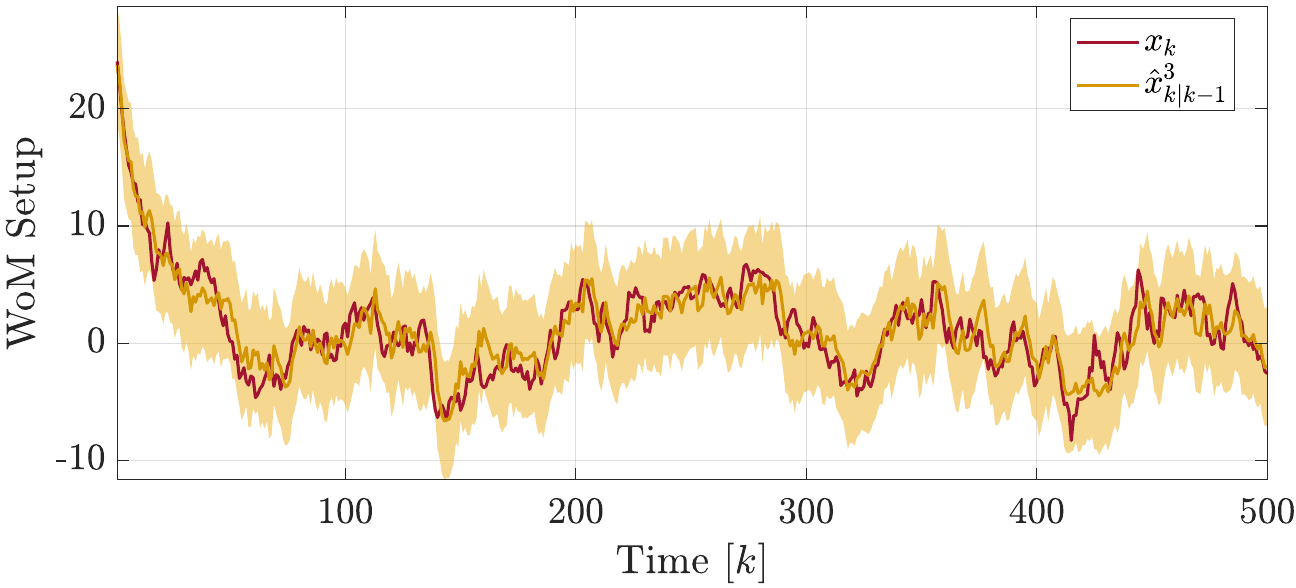}
    \end{subfigure}
    \vspace{-0.4 cm}
    \caption{Convergence of the one-step-ahead predictions to a stationary process in the PP (upper panel) and WoM cases (lower panel).}
    \label{fig:convergence of prediction mean}
\end{figure}


%% file: conclusions.tex

\section{Conclusions} \label{sec:CONCLUSIONS}

In this work, we investigate a specific variant of social learning, namely Word-of-Mouth learning, where a series of Kalman filter agents estimates the state of a dynamical system, with the estimate of the final agent serving as the public belief. We extensively study the stationary properties of the estimates produced by the agents, both theoretically and through simulations. A key finding of our study is that some agents, particularly those closer to the final agent, benefit from using the public belief rather than relying on their private knowledge. In contrast, agents closer to the data-generating mechanism experience performance degradation. 

Extending our analysis to vector-valued states emerges as a natural direction for future work. However, this presents significant challenges due to the non-commutativity of matrix operations, which complicates the expression of the WoM DARE, as well as the determination of its fixed points.


%% file: main.bbl
\newcommand{\noop}[1]{}
\begin{thebibliography}{10}

\bibitem{nikolenko2021synthetic}
S.~I. Nikolenko {\em et~al.}, {\em Synthetic data for deep learning}, vol.~174.
\newblock Springer, 2021.

\bibitem{vikram2003incest}
S.~McLaughlin, V.~Krishnamurthy, and S.~Challa, ``Managing data incest in a distributed sensor network,'' in {\em 2003 IEEE International Conference on Acoustics, Speech, and Signal Processing, 2003. Proceedings.}, vol.~5, pp.~V--269, 2003.

\bibitem{diffusion2024}
H.~Cao, C.~Tan, Z.~Gao, Y.~Xu, G.~Chen, P.-A. Heng, and S.~Z. Li, ``A survey on generative diffusion models,'' {\em IEEE Transactions on Knowledge and Data Engineering}, 2024.

\bibitem{zhao2023survey}
W.~X. Zhao, K.~Zhou, J.~Li, T.~Tang, X.~Wang, Y.~Hou, Y.~Min, B.~Zhang, J.~Zhang, Z.~Dong, {\em et~al.}, ``A survey of large language models,'' {\em arXiv preprint arXiv:2303.18223}, vol.~1, no.~2, 2023.

\bibitem{marchi2024heatdeathgenerativemodels}
M.~Marchi, S.~Soatto, P.~Pratik, and P.~Tabuada, ``Heat death of generative models in closed-loop learning,'' {\em arXiv preprint arXiv:2404.02325}, 2024.

\bibitem{biggio2012poisoning}
B.~Biggio, B.~Nelson, and P.~Laskov, ``Poisoning attacks against support vector machines,'' {\em arXiv preprint arXiv:1206.6389}, 2012.

\bibitem{AO11}
D.~Acemoglu and A.~Ozdaglar, ``Opinion dynamics and learning in social networks,'' {\em Dynamic Games and Applications}, vol.~1, pp.~3--49, 2011.

\bibitem{Ban92}
A.~V. Banerjee, ``A simple model of herd behavior,'' {\em The Quarterly Journal of Economics}, vol.~107, no.~3, pp.~797--817, 1992.

\bibitem{BHW92}
S.~Bikhchandani, D.~Hirshleifer, and I.~Welch, ``A theory of fads, fashion, custom, and cultural change as informational cascades,'' {\em Journal of Political Economy}, vol.~100, no.~5, pp.~992--1026, 1992.

\bibitem{BMS20}
V.~Bordignon, V.~Matta, and A.~H. Sayed, ``Adaptive social learning,'' {\em IEEE Transactions on Information Theory}, vol.~67, no.~9, pp.~6053--6081, 2021.

\bibitem{KH15}
V.~Krishnamurthy and W.~Hoiles, ``Online reputation and polling systems: Data incest, social learning, and revealed preferences,'' {\em IEEE Transactions on Computational Social Systems}, vol.~1, no.~3, pp.~164--179, 2014.

\bibitem{KP14}
V.~Krishnamurthy and H.~V. Poor, ``A tutorial on interactive sensing in social networks,'' {\em IEEE Transactions on Computational Social Systems}, vol.~1, no.~1, pp.~3--21, 2014.

\bibitem{Say14b}
A.~H. Sayed, ``Adaptation, learning, and optimization over networks,'' {\em Foundations and Trends in Machine Learning}, vol.~7, no.~4-5, pp.~311--801, 2014.

\bibitem{Cha04}
C.~Chamley, {\em Rational Herds: Economic Models of Social Learning}.
\newblock Cambridge University Press, 2004.

\bibitem{jain2025interacting}
A.~Jain and V.~Krishnamurthy, ``Interacting large language model agents. bayesian social learning based interpretable models.,'' {\em IEEE Access}, 2025.

\bibitem{Vives01}
X.~Vives, ``How fast do rational agents learn?,'' {\em The Review of Economic Studies}, vol.~60, no.~2, pp.~329--347, 1993.

\bibitem{Viv97}
X.~Vives, ``Learning from others: A welfare analysis,'' {\em Games and Economic Behavior}, vol.~20, no.~2, pp.~177--200, 1997.

\bibitem{slowConvergenceSL}
V.~Krishnamurthy and C.~R. Rojas, ``Slow convergence of interacting kalman filters in word-of-mouth social learning,'' in {\em Proceedings of the 60th Annual Allerton Conference on Communication, Control, and Computing}, pp.~1--6, 2024.

\bibitem{simon2006optimal}
D.~Simon, {\em Optimal State Estimation: Kalman, $H_\infty$, and Nonlinear Approaches}.
\newblock John Wiley \& Sons, 2006.

\bibitem{anderson2012optimal}
B.~D.~O. Anderson and J.~B. Moore, {\em Optimal Filtering}.
\newblock Dover, 2012.

\bibitem{Horn-Johnson-12}
R.~A. Horn and C.~R. Johnson, {\em Matrix Analysis, 2nd Ed.}
\newblock Cambridge University Press, 1985.

\bibitem{Hardt-Price-14}
M.~Hardt and E.~Price, ``The noisy power method: A meta algorithm with applications,'' in {\em Advances in Neural Information Processing Systems (NIPS)}, vol.~27, 2014.

\bibitem{kolmogorov1975introductory}
A.~N. Kolmogorov and S.~V. Fomin, {\em Introductory Real Analysis}.
\newblock Dover, 1975.

\bibitem{Cover-Thomas-06}
T.~Cover and J.~Thomas, {\em Elements of Information Theory, 2nd Ed.}
\newblock Wiley Interscience, 2006.

\end{thebibliography}
